\documentclass[%
reprint,
 nofootinbib,
 amsmath,amssymb,
 aps,
 pre,
 superscriptaddress,
]{revtex4-2}

\usepackage{natbib}
\usepackage[utf8]{inputenc} 
\usepackage[T1]{fontenc}    
\usepackage{hyperref}       
\usepackage{url}            
\usepackage{booktabs}       
\usepackage{amsfonts}       
\usepackage{nicefrac}       
\usepackage{microtype}      
\usepackage{amsthm}
\usepackage{amsmath}
\usepackage{float}
\usepackage{graphicx}
\usepackage{subcaption}
\usepackage{mwe}
\usepackage{xcolor}
\usepackage{diagbox}
\usepackage{amsfonts}
\usepackage{amssymb}
\usepackage{times}
\usepackage{multirow}
\usepackage{fancyhdr,graphicx,amsmath,amssymb}
\usepackage[ruled,vlined]{algorithm2e}
\include{pythonlisting}
\SetKwInput{KwResult}{Input}
\SetKwInput{KwResults}{Output}
\usepackage[toc,page]{appendix}
\usepackage{comment}
\usepackage{ulem}
\usepackage{cancel}
\usepackage{bbm}
\usepackage{graphicx}

\newcommand{\mat}[1]{\mathbf{#1}}

\newtheorem{theorem}{Theorem}

\begin{document}

\preprint{APS/123-QED}

\title{Schr\"{o}dinger PCA: On the Duality between Principal Component Analysis and Schr\"{o}dinger Equation}

\author{Ziming Liu}
\email[]{Ziming Liu: zmliu@mit.edu}
\affiliation{Department of Physics, Massachusetts Institute of Technology, Cambridge, USA}

\author{Sitian Qian}
\affiliation{School of Physics, Peking University, Beijing, China}

\author{Yixuan Wang}
\email[]{Yixuan Wang: roywang@caltech.edu}
\affiliation{Applied and Computational Mathematics, California Institute of Technology, Pasadena, USA}

\author{Yuxuan Yan}
\affiliation{School of Physics, Peking University, Beijing, China}

\author{Tianyi Yang}
\affiliation{School of Physics, Peking University, Beijing, China}

\date{\today}

\begin{abstract}
Principal component analysis (PCA) has been applied to analyze random fields in various scientific disciplines. However, the explainability of PCA remains elusive unless strong domain-specific knowledge is available. This paper provides a theoretical framework that builds a duality between the PCA eigenmodes of a random field and eigenstates of a Schr\"{o}dinger equation. Based on the duality we propose  the \textit{Schr\"{o}dinger PCA} algorithm to replace the expensive PCA solver with a more sample-efficient Schr\"{o}dinger equation solver. We verify the validity of the theory and the effectiveness of the algorithm with numerical experiments.

\end{abstract}

\maketitle


\section{Introduction}\label{sec:intro}

Random fields are prevalent and important in many scientific disciplines, e.g. cosmology~\cite{pen1997generating}, high energy physics~\cite{novak2014determining}, fluid dynamics~\cite{pereira_garban_chevillard_2016}, and material science~\cite{pavliotis2014stochastic,garbuno2020interacting}. Since fields values are usually correlated with neighboring points, a low-dimensional description of random fields is not only possible, but also more computationally efficient. One standard method for dimensionality reduction~\cite{STRICHARTZ198951,tensor_decomp_ml,tensor_kolda,H_Tu_2014,wold1987principal,hyvarinen2000independent,dict_learning,nmf} is principal component analysis (PCA), which attempts to diagonalize the correlation matrix via an orthogonal transformation, whose column vectors are known as eigenmodes.

A random field can be characterized by a two-point covariance function $C(\mat{x},\mat{y})$, so the question is: How is $C(\mat{x},\mat{y})$ related to the eigenmodes $\phi_i(\mat{x})$? The answer is the Karhunen–Loève expansion~\cite{schwab2006karhunen}, an established theory in the statistics community. 
The aim of this paper is to answer this question in the taste of physics by revealing the fact that PCA eigenmodes of random field configurations are (approximately) equivalent to eigenstates of a steady-state Schr\"{o}dinger equation. The aspects of such duality is summarized in TABLE \ref{tab:SE-PCA}.

\begin{table}[htbp]
	\caption{The duality between the PCA problem and the quantum steady-state problem}
	\begin{tabular}{|c|c|}
		\hline
		PCA & Schr\"{o}dinger \\\hline
		Variance $A(\mat{x})$ & Potential $V(\mat{x})$  \\\hline
		Correlation Kernel $\mat{\Sigma}(\mat{x})$ &  Inverse Mass $\mat{\Sigma}_m(\mat{x})$ \\\hline 
		Eigenvalue $\lambda_i$ & Eigen-energy $E_i$ \\\hline
		Eigenvector $\phi_i$ & Eigenstate $\psi_i$ \\\hline
	\end{tabular}
	\label{tab:SE-PCA}
\end{table}

For example, we study a toy 2D random field shown in FIG. ~\ref{fig:GRF_illus}. Many realizations of the random field are drawn from the probabilistic distribution determined by $C(\mat{x},\mat{y})$. PCA eigenmodes of these realizations have interesting structures: the eigenmodes look like quantum wave functions of a 2D harmonic oscillator. In this paper, we argue that this is not a mere coincidence but rather a direct consequence of the duality between the PCA problem and the quantum steady-state problem. Definitions of these quantities in the table will be made precise in the next section.

Besides offering a phyiscal perspective of the well-known PCA algorithm, we can also achieve a significant computational speedup benefiting from the duality.
We can solve a PCA problem with a Schr\"{o}dinger solver (i.e. elliptic equation solver), which is the main idea of our proposed Schr\"{o}dinger PCA algorithm. Numerical experiments show that Schr\"{o}dinger PCA requires $100\times$ fewer anchor points to accurately recover eigenvalues and eigenvectors compared to PCA. In the literature, there have been various methods proposed to accelerate the PCA algorithm, including Hilbert space methods exploiting the structure of RKHS ~\cite{amini2012sampled} and stochastic PDE methods that use solutions of SPDE to approximate Gaussian rathom fields~\cite{lindgren2011explicit}. Our Schr\"{o}dinger PCA method instead studies the case where the variance function could be as general as position-dependent, and the correlation length is relatively small, rendering previous methods inapplicable or inefficient. We notice that ~\cite{Akinduko_2014,spatialpca} are the PCA variants closest to our goals, but they do not point out the duality nor take advantage of the elliptic partial differential equation solvers. Although connections between convolutions and differential operators have been well studied in the mathematics literature~\cite{Bellman1964DifferentialAA, DISTEFANO19701021, DRAGANOV2010952}, we provide a particular case where an intuitive physics picture is available.

\begin{figure}[t]
	\centering
	\includegraphics[width=1.0\linewidth]{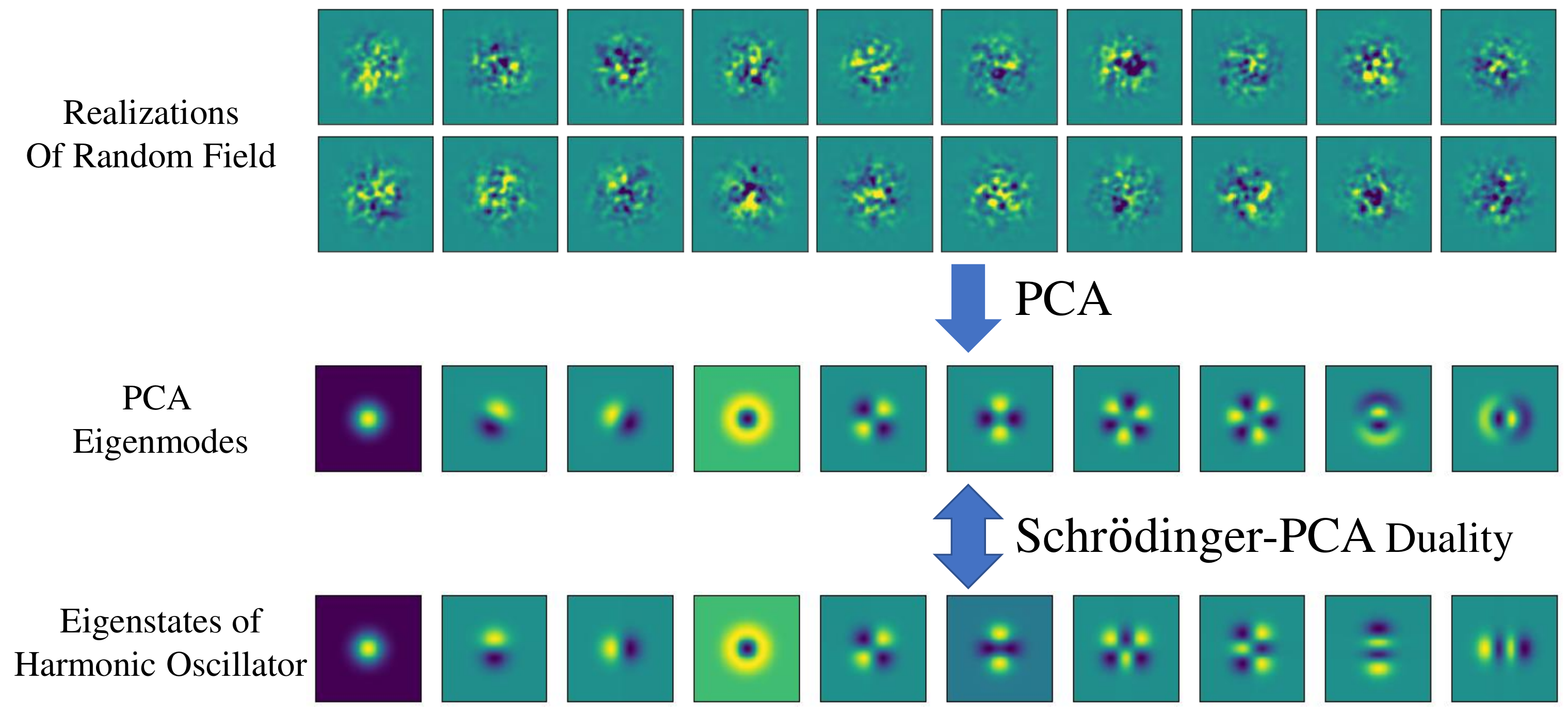}
	\caption{PCA eigenmodes of random field configurations share much similarity with eigenstates of a 2D quantum harmonic oscillator.}
	\label{fig:GRF_illus}
	\vskip -0.5cm
\end{figure}

\section{Method}\label{sec:background}

\subsection{Notations and Backgrounds}

{\bf Random Field}
A random field is a stochastic process where  $\phi:\mathbb{R}^d\to\mathbb{R}$ is a stochastic function, which we refer as \textit{field} below. Two-point covariance $C(\mat{x},\mat{y})$, one-point variance $A(\mat{x})$, and two-point correlation $K(\mat{x},\mat{y})$ of the random field are defined as:
\begin{equation}\label{eq:grf}
\begin{aligned}
	&C(\mat{x},\mat{y})\equiv \langle \phi(\mat{x})\phi(\mat{y})\rangle\\
    &A(\mat{x})\equiv C(\mat{x},\mat{x})= \langle \phi(\mat{x})\phi(\mat{x})\rangle\\
    &K(\mat{x},\mat{y})\equiv C(\mat{x},\mat{y})/\sqrt{A(\mat{x})A(\mat{y})}
\end{aligned}
\end{equation}
where $\langle\cdots\rangle$ means averaging over an ensemble of realizations of the random field or integral over the probability distribution:

\begin{equation}
\begin{aligned}
P(\phi(\mat{x}_1)=v_1,\cdots,&\phi(\mat{x}_n)=v_n)\\
&\propto {\rm exp}(-\sum_{ij}C(\mat{x}_i,\mat{x}_j)v_iv_j)
\end{aligned}
\end{equation}

Usually  $K(\mat{x},\mat{y})=f(|\mat{x}-\mat{y}|)$ where $|\cdot|$ can be an arbitrary norm, $f$ is a monotonically decreasing function and approaches to 0 as $|\mat{x}-\mat{y}|\to\infty$ due to locality. To make things concrete, we focus on Gaussian random fields~\footnote{Our analysis can also extend to laplace-type correlation and possibly many other tailed correlations (please see Appendix B).}, where

\begin{equation}
\begin{aligned}
K(\mat{x},\mat{y})={\rm exp}(-\frac{1}{2}(\mat{y}-\mat{x})^T\mat{\Sigma}^{-1}(\mat{y}-\mat{x}))\\
\end{aligned}
\end{equation}

Without loss of generality, we have assumed $\langle \phi(x)\rangle=0$~\footnote{Principal component analysis will remove the mean value before solving the eigen-problem.}. In the special case where $\mat{\Sigma}=\sigma^2\mat{I}$, $\sigma$ characterizes the correlation length of the Gaussian random field. The correlation length can originate from the smoothness of functions or from physical entities, e.g. the width of a  particle~\cite{miller2007glauber}. A large class of spatial data (including but not limited to physics) can be modeled with Gaussian random fields~\cite{neal1997monte,pen1997generating,novak2014determining,gp_robot,gramacy2008bayesian}, in that Gaussian random fields capture two key features: fluctuations and (local) correlations. In numerical experiments, Gaussian random fields are efficiently generated by smearing out completely uncorrelated random fields with a Gaussian filter (see Appendix D).

{\bf PCA} PCA plays an important role in dimension reduction, pattern recognition, and partial differential equations (PDE)~\cite{zhang2016big,hoffmann2007kernel,perlibakas2004distance,ding2004k,owhadi2015bayesian}. The key idea of PCA is to diagonalize the correlation matrix via an orthogonal transformation. For the random fields described by Eq.~(\ref{eq:grf}), eigenmodes $\phi_i(\mat{x}) (i=1,2,\cdots)$ should satisfy the eigen-equations and orthogonality constraints:
\begin{equation}\label{eq:eigen_pca}
\begin{aligned}
    &\int_{\mat{y}} d\mat{y}C(\mat{x},\mat{y})\phi_i(\mat{y})=\lambda_i\phi_i(\mat{x}), 
    &\int_{\mat{x}} d\mat{x} \phi_i(\mat{x})\phi_j(\mat{x})=\delta_{ij}
\end{aligned}
\end{equation}
where eigenvalues $\lambda_i$ are ordered as $\lambda_1\geq\lambda_2\geq\cdots\geq 0$. The non-negative eigenvalues and existence of orthogonal eigenmodes are due to the positive definiteness of $C(\mat{x},\mat{y})$. Eigen-decomposition of random field is also known is 
Karhunen–Loève expansion in the statistics literature~\cite{schwab2006karhunen}.

{\bf Schr\"{o}inger Equation} The Schr\"{o}dinger equation (SE)~\cite{griffiths2018introduction} uses a wave function $\psi(\mat{x})\ (\mat{x}\in\mathbb{R}^d)$ to characterize a quantum particle. The steady-state SE describes a particle in a potential field $V(\mat{x})$ with constant energy $E$:
\begin{equation}\label{eq:se}
\begin{aligned}
    \hat{H}\psi(\mat{x})=-\frac{\mat{\Sigma}_m(\mat{x})}{2}\Delta \psi(\mat{x})+V(\mat{x})\psi(\mat{x})=E\psi(\mat{x})
\end{aligned}
\end{equation}
where the positive definite matrix $\mat{\Sigma}_m\in\mathbb{R}^{d\times d}$ is the inverse mass matrix of the particle, $\Delta$ is the Laplacian operator: $\Delta=\nabla^2=\sum_{i=1}^d \frac{\partial^2}{\partial x_i^2}$. The first term $-\frac{\mat{\Sigma}_m(\mat{x})}{2}\Delta\psi(\mat{x})$ and the second term $V(\mat{x})\psi(\mat{x})$ correspond to the kinetic and potential energy of a particle, respectively. Because $\hat{H}$ is an Hermitian operator, we have a complete set of real eigenstates $\psi_i(\mat{x})$ that satisfy:
\begin{equation}\label{eq:eigen_se}
\begin{aligned}
    \hat{H}\psi_i(\mat{x})=E_i\psi_i(\mat{x}),\quad \int_{\mat{x}}d\mat{x}\psi_i(\mat{x})\psi_j(\mat{x})=\delta_{ij}
\end{aligned}
\end{equation}
where $\psi_i(\mat{x})$ and $E_i$ are the $i$-th \textit{eigenstate} and \textit{eigen-energy}, respectively.

\subsection{The duality between PCA and Schr\"{o}dinger equation}

We shall build the correspondence between PCA and Schr\"{o}dinger equation under some mild assumptions. Comparing Eq.~(\ref{eq:eigen_pca}) and (\ref{eq:eigen_se}), similar structures can be observed, and intuitively the eigenmodes with the largest eigenvalues of the PCA problem could be mapped one-by-one to the first few lowest energy states of the Schr\"{o}dinger equation. 
We aim to establish the duality in Theorem \ref{thm:1}.

\begin{theorem}\label{thm:1}
(Informal) The PCA problem of Eq.~(\ref{eq:eigen_pca}) can be approximated by the Schr\"{o}dinger problem of Eq.~(\ref{eq:eigen_se}), i.e. two systems have the same equation hence same eigenvalues and eigenmodes up to a (possible) minus sign~\footnote{The eigenvalues of PCA and SE match each other by a minus sign. In terms of eigenmodes, there is a possible minus sign as in most eigen-problems: $\psi_i$ is an eigenstate if and only if $-\psi_i$ is an eigenstate.}, up to a second-order approximation provided that these quantities are equal:
\begin{eqnarray}\label{eq:equivalence}
\begin{aligned}
    \phi(\mat{x})&\Longleftrightarrow\psi(\mat{x})\\ -A(\mat{x})&\Longleftrightarrow V(\mat{x})\\ A(\mat{x})\mat{\Sigma}(\mat{x})&\Longleftrightarrow\mat{\Sigma}_m(\mat{x}).
\end{aligned}
\end{eqnarray}
This approximation is valid for the first few eigenmodes, but not for higher eigenmodes.
\end{theorem}


We summarize the key idea of the proof here. Formal formulation of Theorem \ref{thm:1} and proof details can be found in Appendix A. We characterize a general operator $\mat{M}$ by its Rayleigh quotient ($\Omega$ is the integral domain):
\begin{equation}\label{eq:ray}
\begin{aligned}
    R(\mat{M},\phi)=\frac{\int_{\mat{x}\in\Omega}{d\mat{x}\phi(\mat{x})\mat{M}\phi(\mat{x})}}{\int_{\mat{x}\in\Omega}d\mat{x}\phi^2(\mat{x})}= \int_{\mat{x}\in\Omega}d\mat{x}\phi(\mat{x})\mat{M}\phi(\mat{x})\\(||\phi||^2=\int_{\mat{x}\in\Omega}\phi^2(\mat{x})d\mat{x}=1)
\end{aligned}
\end{equation}
The rayleigh quotient of PCA and SE are defined as $R_{pca}(C,\phi)$ and $R_{se}(\hat{H},\psi)$ respectively:

\begin{equation}\label{eq:ray_pca}
\begin{aligned}
    R_{pca}(C,\phi)&=\int_{\mat{x}\in\Omega,\mat{y}\in\Omega} d\mat{x}d\mat{y} \phi(\mat{x})C(\mat{x},\mat{y})\phi(\mat{y})
    \\
    &=\int_{\mat{x}\in\Omega}d\mat{x}\phi(\mat{x})\{\int_{\mat{y}\in\Omega}d\mat{y}C(\mat{x},\mat{y})\phi(\mat{y})\}\\
\end{aligned}
\end{equation}

\begin{equation}\label{eq:ray_se}
\begin{aligned}
R_{se}(\hat{H},\psi)&=\int_{\mat{x}\in\Omega} d\mat{x} \psi(\mat{x})\hat{H}\psi(\mat{x})\\
&=\int_{\mat{x}\in\Omega}d\mat{x} \psi(\mat{x})\{(-\frac{\mat{\Sigma}_m}{2}\Delta+V(\mat{x}))\psi(\mat{x})\}
\end{aligned}
\end{equation}

The only thing left is to prove the (approximate) equivalence of $\{\cdots\}$ in Eq.~(\ref{eq:ray_pca}) and (\ref{eq:ray_se}) by invoking Guassian integrals and integration by parts (details in Appendix A). To sum up, the approximation is valid for lower eigenmodes in practice, but we don't know a priori how many eigenmodes can be well approximated. Rather, the formal version of Theorem 1 in  Appendix A sheds some light on how to post-process and verify a posteriori that the approximation is valid.

Although highly speculative, this duality between PCA and the quantum steady-state problem might shed light on deep theoretical links between quantum mechanics (Schr\"{o}dinger) and information compression (PCA), which will be explored in future works.

\subsection{Discussions}

{\bf Implications of the Duality}
The duality also explains scientists' conventions when they analyze a signal/field by decomposing it into a combination of hand-designed modes, e.g. Fourier bases, spherical harmonics, Bessel functions, etc. We argue that scientists are implicitly doing some optimal compression (PCA) when choosing the hand-designed modes. For simplicity, we consider a uniform random field such that $\mat{\Sigma}(\mat{x})$ and $A(\mat{x})$ are independent of $\mat{x}$. These conditions are then translated to uniform potential energy $V(\mat{x})=V_0$ (WLOG $V(\mat{x})=0$) and constant mass in the Schr\"{o}dinger equation. Consequently, the Schr\"{o}dinger equation describes a free particle on the domain $\Omega$, where both the random field and Schr\"{o}dinger equation are defined on.

\begin{table}[htbp]
	\caption{`Optimal' bases for different problems}
	\begin{tabular}{|c|c|c|}
		\hline
		Domain $\Omega$ & Eigenstates/Bases & Physics Example \\\hline
		Sphere (2D) & Spherical Harmonics & CMB in Cosmology \\\hline
		Circular Well (2D) & Bessel Function & Circular Waveguide \\\hline
		Periodic Box ($n$D) & Fourier Bases & Fluid Dynamics \\\hline
	\end{tabular}
	\label{tab:examples}
\end{table}

Suppose $\Omega$ is a sphere (2D). On the one hand, eigenstates of a free particle on a sphere are spherical harmonics, i.e. eigenstates of the angular momentum operator. On the other hand in cosmology, spherical harmonics are leveraged to analyze the (approximately) isotropic cosmic microwave background (CMB)~\cite{durrer2020cosmic}. We speculate this is not coincident but a direct consequence of the duality. If we humans are living in a universe where CMB is largely anisotropic, the spherical harmonics are no longer `optimal' in the sense of PCA (see the anisotropic example in Section \ref{sec:exp_climate}). Similarly, a particle trapped in a deep circular well (2D) has Bessel functions as eigenstates, a free particle in a periodic box has plane wave solutions (i.e. Fourier). These basis functions are used in waveguide design and fluid dynamics respectively. We summarize the relations in TABLE \ref{tab:examples}. 

{\bf No Quantum Tunneling for PCA} Quantum mechanics is distinguished from classical mechanics by its uncertainty principle, leading to the intriguing phenomenon of quantum tunneling. When a particle (with state-independent mass) is trapped in a finite deep potential well i.e. $V(\mat{x})=-V_0<0 (\mat{x}\in\Omega)$ and $V(\mat{x})=0 (\mat{x}\not\in\Omega)$, the wave function of the particle is non-zero outside $\Omega$ even when the particle has negative energy $-V_0<E<0$. However, there does not exist a random field model dual to the quantum system. According to Theorem \ref{thm:1}, one-point variance is determined $A(\mat{x})=V_0>0 (\mat{x}\in\Omega)$ and $A(\mat{x})=0 (\mat{x}\not\in\Omega)$. For $\mat{x}\not\in\Omega$, inverse mass matrix $\mat{\Sigma}_m(\mat{x})=A(\mat{x})\mat{\Sigma}(\mat{x})=0$ i.e. an infinite mass. The consequence is : Outside $\Omega$ the particle is infinitely heavy thus cannot move at all, i.e. quantum tunneling is impossible. This actually makes perfect sense for PCA: if a field value never fluctuates, it should contribute exactly zero to all eigenmodes.

{\bf Extension to Laplace-type random field}  The duality can be similarly extended to the Laplace-type kernel, with only coefficients different from the Gaussian case. Please refer to Appendix B and C for details.

\section{Schr\"{o}dinger PCA Algorithm}

We aim to design an algorithm, dubbed as Schr\"{o}digner PCA, to convert a PCA problem to the corresponding Schr\"{o}dinger problem. The main idea is already clear in the aforementioned duality. We start this section by revealing the limitation of PCA, motivating Schr\"{o}dinger PCA.

\subsection{Motivation and Algorithm}

{\bf Why PCA Degrades} In practice, only a finite set of \textit{anchor points} $S=\{\mat{x}_1,\cdots,\mat{x}_n\}$ are available for field measurements i.e. only $(\phi(\mat{x}_1),\cdots,\phi(\mat{x}_n))$ are accessible rather than $\phi(\mat{x})$ for any $\mat{x}\in\Omega\subset\mathbb{R}^d$. The procedure of PCA on discrete anchor points is shown in Alg.~{\ref{alg:pca}}. In brief, one applies PCA to many realizations of the field values at anchor points, obtain eigenmodes defined on these anchor points and interpolate eigenvectors to continuous functions.

In the case $\mat{\Sigma}=\sigma^2\mat{I}_{d\times d}$ where $\sigma$ is the correlation length, two neighboring anchor points should be close enough (e.g. distance $\lesssim\sigma$) to make sure the covariance is significant enough to reveal local correlation. As a rough estimate, we suppose $\Omega$ has size $L$ on each dimension, then the number of anchor points required scales as $(\frac{L}{\sigma})^{d}$ which increases as $\frac{L}{\sigma}$ increases (multi-scale) or $d$ increases (curse of dimensionality). We will attack the multi-scale issue with Schr\"{o}dinger PCA algorithm. 

{\bf Schr\"{o}dinger PCA} Instead of trying to detect local correlation with densely distributed anchor points, Schr\"{o}dinger PCA bypasses this multi-scale issue by intergrating out local Gaussian correlation analytically thus transforming the PCA problem to a Schr\"{o}dinger equation. The key steps are summarized here (also in Alg.~\ref{alg:sepca}): (1) compute only \textit{variance} of fields at each anchor point; (2) leveraging dualities in Eq.~(\ref{eq:equivalence}) and interpolate to obtain the continuous potential function $V(\mat{x})$ and inverse mass matrix $\mat{\Sigma}_m(\mat{x})$; (3) Solving the Schr\"{o}dinger equation Eq.~(\ref{eq:se}) with elliptic equation solvers.

\subsection{Discussions}

{\bf When Schr\"{o}dinger PCA Outperforms PCA} It is worth emphasizing the key factor that contributes to the success of  Schr\"{o}dinger PCA lies in scale separation, i.e. $\phi(\mat{x})$ have smaller scales than $A(\mat{x})$ such that $\nabla\nabla^T A(\mat{x})/A(\mat{x})\ll\nabla\nabla^T\phi(\mat{x})/\phi(\mat{x})$. While PCA requires anchor points to characterize $\phi(\mat{x})$ (small scale), Schr\"{o}dinger PCA only requires anchor points to characterize $A(\mat{x})$ (large scale). For this scale separation to work, Schr\"{o}dinger PCA requires extra estimation of the small-scale correlation kernel $K(\mat{x},\mat{y})$ in order to determine the coefficients in the Schr\"{o}dinger equation, but this is cheap because either (1) the correlation kernel is known a prior (e.g. can be derived from physical principles); (2) a few \textit{detector points} can be added locally around each anchor point to estimate the correlation kernel. In the isotropic case $\mat{\Sigma}=\sigma^2\mat{I}$, only one detector point is required to determine $\sigma$ so the extra computational costs can be ignored. For simplicity, we choose $\mat{\Sigma}=\sigma^2\mat{I}$ and assume $\sigma$ is known a priori (but can depend on $\mat{x}$) in the following numerical experiments.

{\bf Error analysis of the second-order approximation of Schr\"{o}dinger PCA}: Since the duality between PCA and Schr\"{o}dinger equation is only approximated up to second-order approximations, higher-order terms should be small enough for this duality to be valid. The  only two approximations used in proof of Theorem \ref{thm:1} (Appendix \ref{app:gaussian_theory}) are: (1) The derivative of $A(\mat{x})$ is truncated to the first order $\nabla A(\mat{x})$, thus higher-order terms $\nabla^n A(\mat{x}) (n\geq 2)$ can be neglected; (2)  The derivative of $\phi(\mat{x})$ is truncated to the second order $\nabla\nabla \phi(\mat{x})$, thus higher-order terms $\nabla^n \phi(\mat{x}) (n\geq 3)$ can be neglected. In Section \ref{sec:exp} we show in practice, the errors can be almost neglected for leading eigenmodes/eigenvalues. We leave the full analysis of how perturbation of high-order terms influence eigenmodes/eigenvalues in future works.

\begin{algorithm}[t]
\SetAlgoLined
\KwResult{Data points $\mat{X}=\{\mat{x}_i\in\mathbb{R}^d\}$ and corresponding field value $\mat{\Phi}=\{\phi(\mat{x}_i,t)\in\mathbb{R}\}$ $(i=1,\cdots,m; t=1,\cdots,n)$ where $i$ labels different anchor points in space, and $t$ refers to different time of measurements or realizations. We have removed the mean value of fields such that $\frac{1}{n}\sum_{t=1}^n \phi(\mat{x}_i,t)=0$ for all $i$.}
(1) Compute \textit{covariance} matrix $\mat{C}_{ij}={\rm Cov}(\phi(\mat{x}_i),\phi(\mat{x}_j))=\frac{1}{n}\sum_{t=1}^n \phi(\mat{x}_i,t)\phi(\mat{x}_j,t)$\;
(2) Diagonalization of $\mat{C}$ such that $\mat{C}=\sum_k \lambda_k\phi_k\phi_k^T$ where $\lambda_k$ and $\phi_k$ are the $k$-th eigenvalue and eigenmode\;
(3) Interpolation of finite-dimensional eigenmode $\phi_k\in\mathbb{R}^m (k=1,2,\cdots)$ to a function $\phi_k(\mat{x})$ that satisfies $\phi_k(\mat{x}_l)=\phi_{k,l}$ where $\phi_{k,l}$ is the $l$-th entry of $\phi_k$ which corresponds to the field value at $\mat{x}_l$ in mode $\phi_k$\;
\KwResults{$\{\lambda_k,\phi_k(\mat{x})\}(k=1,2,\cdots)$}
\caption{Principal Component Analysis for Mode Decomposition}
\label{alg:pca}
\end{algorithm}

\begin{algorithm}[t]
\SetAlgoLined
\KwResult{Data points $\mat{X}=\{\mat{x}_i\in\mathbb{R}^d\}$, field value $\mat{\Phi}=\{\phi(\mat{x}_i,t)\in\mathbb{R}\}$ and estimated correlation kernel $\mat{\Sigma}=\{\mat{\Sigma}(\mat{x}_i)\in\mathbb{R}^{d\times d}\}$ $(i=1,\cdots,m; t=1,\cdots,N)$}
(1) Compute \textit{variance} vector: $A(\mat{x}_i)=\frac{1}{N}\sum_{t=1}^N|\phi(\mat{x}_i,t)|^2$\;
(2) Define corresponding potential vector as $V_i=-A(\mat{x}_i)$ and interpolation of $V\in\mathbb{R}^m$ to a function $V(\mat{x})$ satisfying $V(\mat{x}_i)=V_i$\;
(3) Define corresponding inverse mass matrix as $\mat{\Sigma}_{m,i}=A(\mat{x}_i)\mat{\Sigma}(\mat{x}_i)$ and interpolation of $\mat{\Sigma}_m\in\mathbb{R}^{d\times d}$ to a matrix function $\mat{\Sigma}_m(\mat{x})$ satisfying $\mat{\Sigma}_m(\mat{x}_i)=\mat{\Sigma}_{m,i}$\;
(4) Solve the schr\"{o}dinger equation $-\frac{\mat{\Sigma}_m(\mat{x})}{2}\Delta \psi(\mat{x})+V(\mat{x})\psi(\mat{x})=E\psi(\mat{x})$ and obtain eigenstates $\psi_k(\mat{x})$ and corresponding energies $E_k$\;
\KwResults{$\{E_k,\psi_k(\mat{x})\}(k=1,2,\cdots)$}
\caption{Schr\"{o}dinger Principal Component Analysis for Mode Decomposition}
\label{alg:sepca}
\end{algorithm}

\section{Numerical Experiments}\label{sec:exp}
Firstly, we revisit the example mentioned in FIG. \ref{fig:GRF_illus} to test the correctness of the correspondence between PCA and Schr\"{o}dinger equation stated in Theorem \ref{thm:1}. We find PCA and Schr\"{o}dinger PCA have similar eigenvalues and eigenmodes for the fine grid (i.e. a large number of anchor points), but PCA fails for the coarse grid (i.e. a small number of anchor points) while Schr\"{o}dinger PCA can still obtain accurate eigenvalues and eigenmodes. Secondly, we apply our method to the global climate example which is modeled as a random field on a 2D sphere $S^2$. The climate modes discovered by Schr\"{o}dinger PCA admit a nice physical interpretation and demonstrate the potential of our method to attack problems on  graphs and manifolds beyond the Euclidean space.

\subsection{Two-dimensional Gaussian Process}

The two-dimensional Gaussian random field lives on the grid [-50,50]$\times$[-50,50] (size $L\times L\equiv100\times 100$). The field generation process can be found in Appendix D. The generated profiles are denoted as $\phi_t(\mat{x})$ where $\mat{x}=(x_1,x_2)$ and $t=1,\cdots,N$ index to distinguish among different realizations. The random field satisfies these statistical properties:
\begin{align}\label{eq:exp_realization}
& A(\mat{x}) =
\left\{
    \begin{array}{cr}
    1-(x_1^2+x_2^2)/40^2 & (x_1^2+x_2^2\leq 40^2)  \\
    0 & (x_1^2+x_2^2> 40^2)\\
    \end{array}
\right.\\
& C(\mat{x},\mat{y}) = \sqrt{A(\mat{x})A(\mat{y})}{\rm exp}(-\frac{(\mat{x}-\mat{y})^T(\mat{x}-\mat{y})}{2\sigma^2})\quad  (\sigma=3)
\end{align}

Note that $\sigma=3\ll L=100$, this corresponds to a multi-scale scenario where Schr\"{o}dinger PCA obtains more accurate results than PCA in the undersampling regime.

\subsubsection{Oversampling regime with fine grids}
In this part, we use numerical results to verify the equivalence between PCA and Schr\"{o}dinger equation, as pointed out in Theorem \ref{thm:1}. PCA (fine grid) is treated as the gold standard method and its results are treated as ground truths.

{\bf PCA (fine grid)}: We use all $101\times 101=10201$ points as anchor points and generate 40000 realizations of Gaussian random field satisfying Eq.~(\ref{eq:exp_realization}). We apply PCA to obtain the first $k=21$ eigenmodes and corresponding eigenvalues. The distance of two neighboring anchor points is 1 (smaller than $\sigma=3$). In this sense, the anchor points oversample the random field, so covariance information among different spatial points can be captured. We implement PCA by using {\tt sklearn.decomposition.pca}.

{\bf SE (fine grid)}: We use all 10201 anchor points to evaluate $A(\mat{x})$ to obtain the potential function $V(\mat{x})$ and inverse mass matrix $\mat{\Sigma}_m(\mat{x})$ as indicated in Theorem \ref{thm:1}, then the Schr\"{o}dinger equation Eq.~(\ref{eq:se}) is discretized on the same grid with the finite difference method. The eigenvalue problem is solved by the python eigen-solver named {\tt numpy.linalg.eigh}. We collect first $k=21$ eigenstates and eigenenergies. These eigenenergies are negative and are negated below to compare with PCA results.

{\bf SE (harmonic approximation)}: The potential function $V(\mat{x})=-A(\mat{x})$ has a quadratic form that is reminiscent of the harmonic oscillator in quantum physics. Although $\mat{\Sigma}_m(\mat{x})$ depends explicitly on $\mat{x}$ it can be nicely approximated as $\mat{\Sigma}_m(\mat{0})$ for ground states and low excited states (eigenstates with low energies) because the wavefunction $\psi(\mat{x})$ centered around $\mat{x}=0$ where the potential energy is global minima. As soon as the system is identified as a harmonic oscillator, the eigenenergies and eigenstates are immediately accessible via our knowledge of quantum mechanics, detailed in Appendix E.

As shown in FIG. ~\ref{fig:eigenvalues}, three methods listed above agree well in terms of eigenvalues for first $k=10$ modes. For $k>10$ modes, the staircase degeneracy structure still remains the same for all methods above, while the values have deviations $\sim 0.06$, which is because the approximation $\mat{\Sigma}_m(\mat{x})\approx\mat{\Sigma}_m(\mat{0})$ becomes poorer for states with higher energies. The staircase degeneracy structure can be elegantly understood by the energy spectrum of 2-d quantum harmonic oscillator: an energy level with quantum number $n\geq 0$ admits $n+1$ ways to choose two integers $n_1,n_2\geq0$ such that $n_1+n_2=n$. Likewise, eigenmodes for all methods are illustrated and compared in FIG. ~\ref{fig:modes} and they show similar behavior for the three aforementioned methods.

As mentioned in the introduction, the Stochastic PDE (SPDE) method and the  Reproducing Kernel Hilbert Space (RKHS) method can also be used to solve the eigen-decomposition of random fields. We show below that neither methods are proper to attack our particular case with position-dependent variance and small correlation length.

{\bf SPDE}: The Matern-type random field has eigenmodes which are solutions of the following equality~\cite{lindgren2011explicit}:

\begin{equation}\label{eq:spde}
    (\kappa^2-\Delta)^{\alpha/2}\phi_i(x)=\lambda_i\phi_i(x)
\end{equation}

Although Eq.~(\ref{eq:spde}) looks similar to Eq.~(\ref{eq:se}), our toy random field is not the Matern-type random field, so Eq.~(\ref{eq:spde}) is inapplicable to our example. The Matern-type variance function should be at least position-independent, but our variance function is indeed position-dependent. Consequently we do not expect SPDE to work, but for illustration we choose $\alpha=2$ and $\kappa=1/\sigma=1/3$. Eq.~(\ref{eq:spde}) therefore becomes a wave equation whose eigenmodes are 2D Fourier bases with wave-number $(k_1,k_2)=\frac{2\pi}{L}(n_1,n_2)$. And according to the wave-number spectrum obtained in~\cite{lindgren2011explicit}, eigenvalues are 
\begin{equation}
    \lambda_i = C(\kappa^2+k_{i,1}^2+k_{i,2}^2)^{-\alpha}
\end{equation}
Eigenvalues and eigenmodes obtained from SPDE are plotted in FIG.~\ref{fig:eigenvalues} and FIG.~\ref{fig:modes}.

{\bf Reproducing Kernel Hilbert Space (RKHS)}: The space of bandlimited continuous functions (truncated Fourier space) is a RKHS~\cite{amini2012sampled}. We first project random fields onto the leading 100 Fourier bases (i.e. truncated Fourier transformation), obtaining Fourier coefficients, followed by applying PCA to those Fourier coefficients. Eigenmodes of Fourier coefficients are then transformed back via an inverse Fourier transformation. The RKHS method is insufficient to reproduce eigenvalues and eigenmodes in FIG.~\ref{fig:eigenvalues} and FIG.~\ref{fig:modes} because our toy problem has small correlation length, and the leading 100 Fourier bases cannot capture such small-scale information.

\subsubsection{Undersampling regime with coarse grids}
When the number of anchor points decreases, PCA will fail in terms of both eigenvalues and eigenmodes; by contrast, Schr\"{o}dinger PCA remains quite robust and behaves nearly the same as the Schr\"{o}dinger PCA in the oversampling regime, as illustrated in FIG. ~\ref{fig:eigenvalues} and \ref{fig:modes}.

{\bf PCA (coarse grid 1)}: We set anchor points on the uniform coarse grid $[-50,40,\cdots,50]^2$ (121 anchor points in total). The same 40000 random field realizations are used as in the oversampling test, but only the field values on the coarse grid are available. We apply PCA to obtain eigenmodes and eigenvalues defined on the coarse grid, and interpolate the eigenmodes back to the fine grid.

{\bf PCA (coarse grid 2)}: Same with {\bf PCA (coarse grid 1)} except that switching the order of PCA and interpolation.

{\bf SE (coarse grid)}: We evaluated $A(\mat{x})$ at anchor points (coarse grid) and interpolate $A(\mat{x})$ back to the fine grid. Correspondingly the values of $V(\mat{x})$ and $\mat{\Sigma}_m(\mat{x})$ can be determined on the fine grid based on Eq.~(\ref{eq:equivalence}). Finally the SE is discretized with the finite difference method and solved with the {\tt eigh} eigen-solver on the fine grid.

\begin{figure}[htbp]
    \centering
    \includegraphics[width=0.8\linewidth]{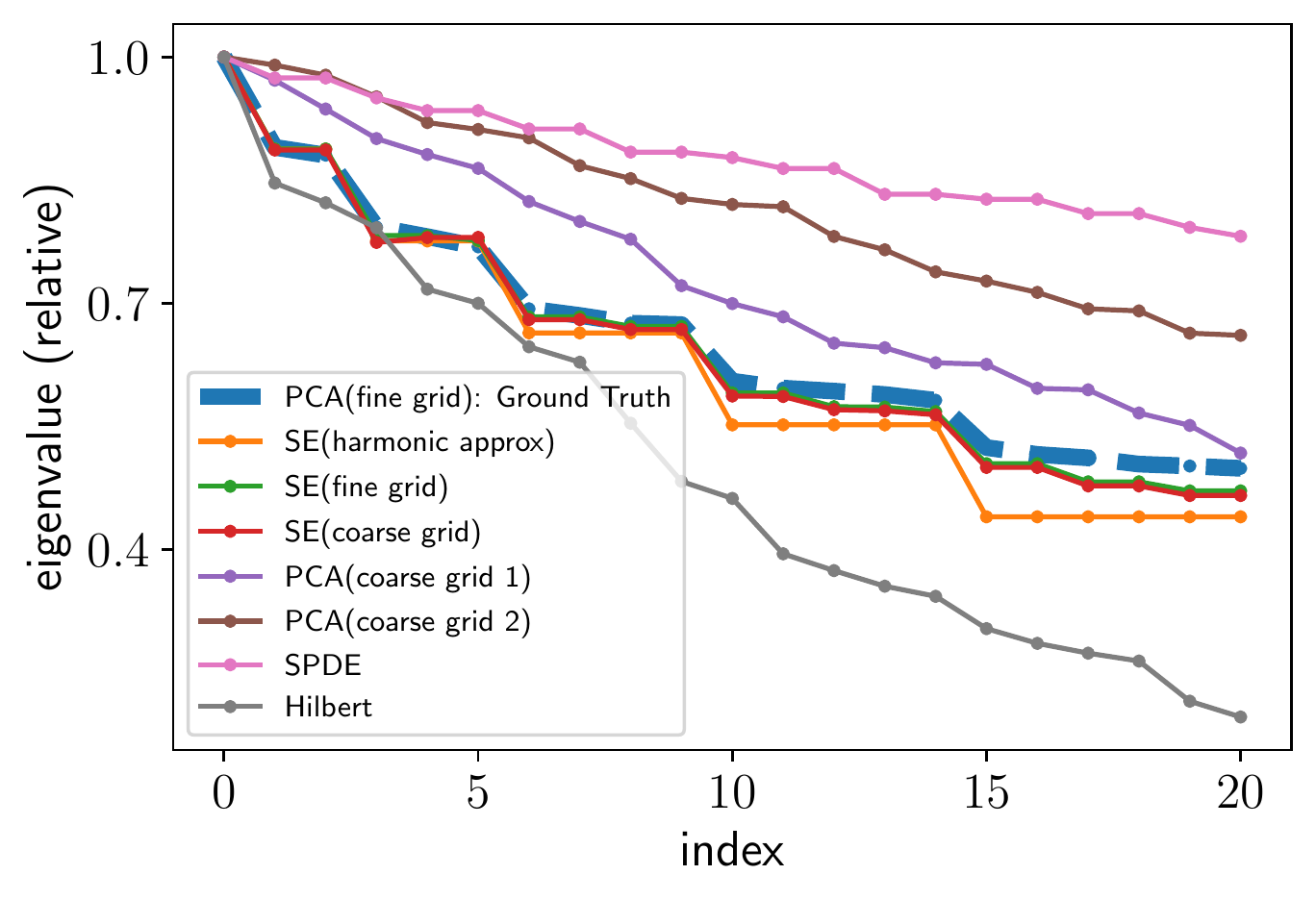}
    \caption{The (relative) eigenvalues of first 21 modes with different algorithms. Schr\"{o}dinger PCA is able to recover eigenvalues on both the fine grid and the corase grid, while other methods either fail (SPDE and Hilbert), or succeed only for the fine grid (PCA).}
    \label{fig:eigenvalues}
\end{figure}

\begin{figure}[t]
    \centering
    \includegraphics[width=1.0\linewidth]{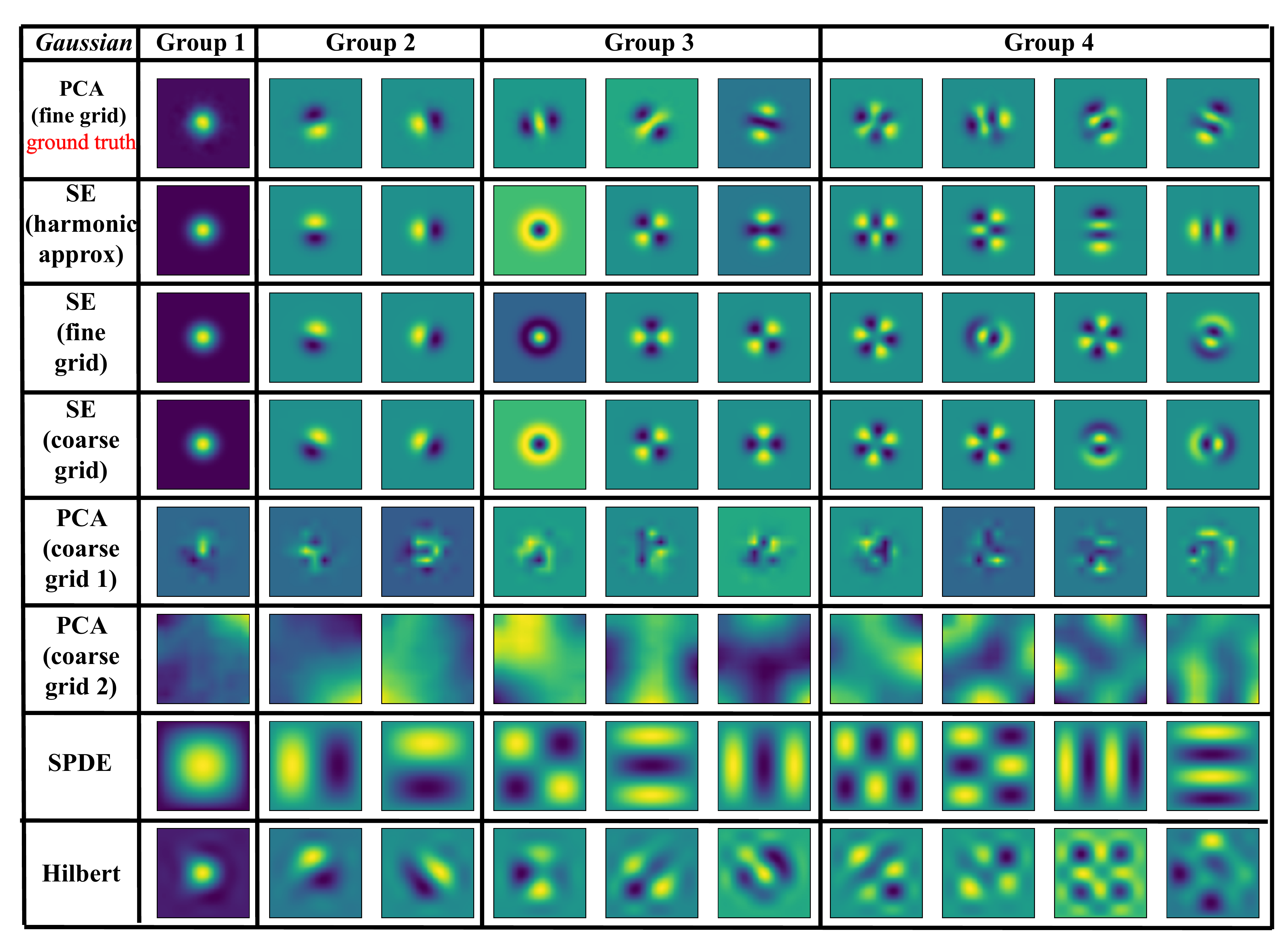}
    \caption{The eigenmodes of first 10 modes from different algorithms. Each square subfigure is a 2D heat diagram where bright (yellow) and dark (dark blue) colors represent high and low values respectively.  Schr\"{o}dinger PCA is able to recover eigenmodes on both the fine grid and the corase grid, while other methods either fail (SPDE and Hilbert), or succeed only for the fine grid (PCA).}
    \label{fig:modes}
\end{figure}

\subsection{Global Climate Modes}\label{sec:exp_climate}
In this example, we demonstrate the effectiveness of the proposed Schr\"{o}dinger PCA to solve a mode decomposition problem on a unit sphere $S^2$. 

Our earth has various kinds of climates, depending on longitude and latitude of the location. In the following, we use spherical coordinate $\mat{x}=(\theta,\varphi)$ and $\phi(\mat{x})$ refers to any climate-related scalar field (e.g. temperature). Both fluctuation magnitude $A(\mat{x})$ and (local) correlation length $\sigma(\mat{x})=\sqrt{\mat{\Sigma}(\mat{x})}$ depends explicitly on $\mat{x}$: (1) $A(\mat{x})$ is small for a place where it is like spring all the year around, while large for another place with scorching summers and cold winters; (2) $\sigma(\mat{x})$ is large in oceans due to ocean circulations, while small on a land without any winds.

We consider two ``earths": isotropic earth (below $a=0$) and anisotropic earth ($a=3$):

\begin{align}\label{eq:climate_realization}
\left\{
    \begin{array}{cr}
    A(\theta,\varphi) = \frac{1}{20}(3 + a{\rm cos}\theta)  \\
    \sigma^2(\theta,\varphi) = \frac{1}{50}(3 - a{\rm cos}\theta)d^2 \\
    \end{array}
\right.
\end{align}
To discretize the Laplacian oeprator on a graph, we utilze icosahedron mesh which contains 2562 vertices, 7680 edges and 5120 faces. Here $d$ is the (averaged) length of all edges. The details of generating the mesh can be found in Appendix F. We repalce the Laplacian operator in Eq.~(\ref{eq:se}) with the graph Laplacian matrix and obtain the eigenmodes on the graph. Finally we interpolate eigenmodes back to the sphere and show them (top view) in FIG. ~\ref{fig:climate}. Here $k=0,1,2\cdots$ refers to the index of eigenstates, ordered from lowest energy to highest energy. The isotropic case provides a baseline for spherical function decomposition, and in fact they correspond to spherical harmonic functions. We observe that eigenmodes of the anisotropic earth can shed light on the the fluctuation pattern: (1) Fluctuations are large around the poles, so the first few patterns ($k=1,5,10,20$) only concentrate around poles; (2) intermediate patterns ($k=100,200$) are particularly interesting because such patterns have similar magnitude around the pole and the equator, but finer structure is observed around the pole revealing that the correlation length is smaller around the pole than the equator; (3) the last few patterns $(k=1000,2000)$ capture local fluctuations around the equator, not revealing collective behavior of global temperature~\footnote{Although we do not provide PCA results here, the averaged edge length $d$ is much greater than the (averaged) correlation length $\sigma$, which in principle prevents PCA from working.}. These observations are particularly interesting because it allows one to infer the random field model from eigenmode structures. We will investigate quantitatively the question of `Can one hear the shape of random fields?'~\footnote{Similar to the spirit of `Can one hear the shape of a drum?' ~\cite{10.2307/2313748}.} in future works.

\begin{figure}[htbp]
    \centering
    \includegraphics[width=1.0\linewidth]{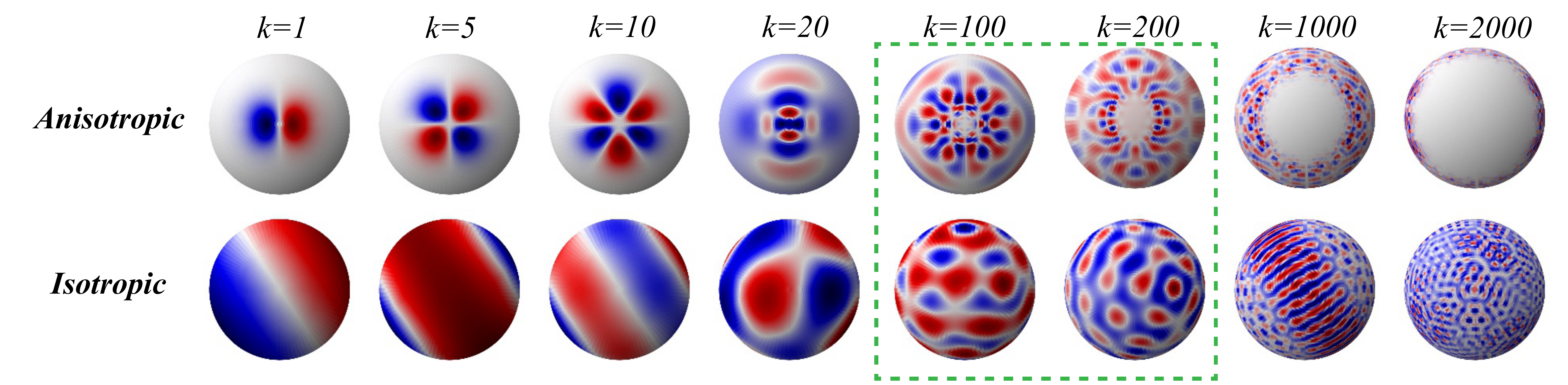}
    \caption{Eigenmodes of the global climate model obtained from Schr\"{o}dinger PCA (top view). First row: Anisotropic globe; Second row: Isotropic globe. Each globe is a heat diagram where bright (red) and dark (blue) colors represent high and low values respectively.}
    \label{fig:climate}
\end{figure}

\section{Conclusions and discussions}\label{sec:conclusion}

In this paper, we first build the duality between PCA eigenmodes of a random field and eigenstates of a Schr\"{o}dinger equation. Based on this observation, we propose the Schr\"{o}dinger PCA algorithm which is more sample-efficient than PCA when the random fields have multi-scale structures. Besides the numerical examples in the previous section, the proposed algorithm has the potential to speed up many scientific computations: (1) Model reduction of turbulent fluid dynamics~\cite{pereira_garban_chevillard_2016}; (2) Analysis of the cosmic microwave background~\cite{pen1997generating}; (3) Collective behavior of relativistic heavy-ion collisions~\cite{novak2014determining}. All of these scenarios have a very narrow correlation kernel, which can be nicely handled by our algorithm, while vanilla PCA will fail on these cases in the undersampling regime.

In the future, it would be interesting to investigate deeper physical implications of the duality between PCA and the Schr\"{o}dinger equation. Moreover, we would like to extend Schr\"{o}dinger PCA algorithm to solve higher-dimensional problems, or problems on graphs and Riemannian manifolds, with the help of state-of-the-art elliptic partial differential equation solvers. 

\section*{Acknowledgement}
We would like to thank Zhihan Li, Yifan Chen, Huichao Song, Houman Owhadi and Max Tegmark for valuable discussions.

\bibliographystyle{unsrt}
\bibliography{ref}

\begin{thebibliography}{10}

\bibitem{pen1997generating}
Ue-Li Pen.
\newblock Generating cosmological gaussian random fields.
\newblock {\em The Astrophysical Journal Letters}, 490(2):L127, 1997.

\bibitem{novak2014determining}
John Novak, Kevin Novak, Scott Pratt, Joshua Vredevoogd, CE~Coleman-Smith, and
  RL~Wolpert.
\newblock Determining fundamental properties of matter created in
  ultrarelativistic heavy-ion collisions.
\newblock {\em Physical Review C}, 89(3):034917, 2014.

\bibitem{pereira_garban_chevillard_2016}
Rodrigo~M. Pereira, Christophe Garban, and Laurent Chevillard.
\newblock A dissipative random velocity field for fully developed fluid
  turbulence.
\newblock {\em Journal of Fluid Mechanics}, 794:369–408, 2016.

\bibitem{pavliotis2014stochastic}
Grigorios~A Pavliotis.
\newblock {\em Stochastic processes and applications: diffusion processes, the
  Fokker-Planck and Langevin equations}, volume~60.
\newblock Springer, 2014.

\bibitem{garbuno2020interacting}
Alfredo Garbuno-Inigo, Franca Hoffmann, Wuchen Li, and Andrew~M Stuart.
\newblock Interacting langevin diffusions: Gradient structure and ensemble
  kalman sampler.
\newblock {\em SIAM Journal on Applied Dynamical Systems}, 19(1):412--441,
  2020.

\bibitem{STRICHARTZ198951}
Robert~S. Strichartz.
\newblock Harmonic analysis as spectral theory of laplacians.
\newblock {\em Journal of Functional Analysis}, 87(1):51 -- 148, 1989.

\bibitem{tensor_decomp_ml}
N.~D. {Sidiropoulos}, L.~{De Lathauwer}, X.~{Fu}, K.~{Huang}, E.~E.
  {Papalexakis}, and C.~{Faloutsos}.
\newblock Tensor decomposition for signal processing and machine learning.
\newblock {\em IEEE Transactions on Signal Processing}, 65(13):3551--3582,
  2017.

\bibitem{tensor_kolda}
Tamara~G. Kolda and Brett~W. Bader.
\newblock Tensor decompositions and applications.
\newblock {\em SIAM Review}, 51(3):455--500, 2009.

\bibitem{H_Tu_2014}
Jonathan H.~Tu, Clarence W.~Rowley, Dirk M.~Luchtenburg, Steven L.~Brunton, and
  J.~Nathan~Kutz.
\newblock On dynamic mode decomposition: Theory and applications.
\newblock {\em Journal of Computational Dynamics}, 1(2):391–421, 2014.

\bibitem{wold1987principal}
Svante Wold, Kim Esbensen, and Paul Geladi.
\newblock Principal component analysis.
\newblock {\em Chemometrics and intelligent laboratory systems}, 2(1-3):37--52,
  1987.

\bibitem{hyvarinen2000independent}
Aapo Hyv{\"a}rinen and Erkki Oja.
\newblock Independent component analysis: algorithms and applications.
\newblock {\em Neural networks}, 13(4-5):411--430, 2000.

\bibitem{dict_learning}
I.~{Tošić} and P.~{Frossard}.
\newblock Dictionary learning.
\newblock {\em IEEE Signal Processing Magazine}, 28(2):27--38, 2011.

\bibitem{nmf}
Daniel~D. Lee and H.~Sebastian Seung.
\newblock Algorithms for non-negative matrix factorization.
\newblock In T.~K. Leen, T.~G. Dietterich, and V.~Tresp, editors, {\em Advances
  in Neural Information Processing Systems 13}, pages 556--562. MIT Press,
  2001.

\bibitem{schwab2006karhunen}
Christoph Schwab and Radu~Alexandru Todor.
\newblock Karhunen--lo{\`e}ve approximation of random fields by generalized
  fast multipole methods.
\newblock {\em Journal of Computational Physics}, 217(1):100--122, 2006.

\bibitem{amini2012sampled}
Arash~A Amini, Martin~J Wainwright, et~al.
\newblock Sampled forms of functional pca in reproducing kernel hilbert spaces.
\newblock {\em The Annals of Statistics}, 40(5):2483--2510, 2012.

\bibitem{lindgren2011explicit}
Finn Lindgren, H{\aa}vard Rue, and Johan Lindstr{\"o}m.
\newblock An explicit link between gaussian fields and gaussian markov random
  fields: the stochastic partial differential equation approach.
\newblock {\em Journal of the Royal Statistical Society: Series B (Statistical
  Methodology)}, 73(4):423--498, 2011.

\bibitem{Akinduko_2014}
A~A Akinduko and A~N Gorban.
\newblock Multiscale principal component analysis.
\newblock {\em Journal of Physics: Conference Series}, 490:012081, mar 2014.

\bibitem{spatialpca}
Urška Demšar, Paul Harris, Chris Brunsdon, A.~Stewart Fotheringham, and Sean
  McLoone.
\newblock Principal component analysis on spatial data: An overview.
\newblock {\em Annals of the Association of American Geographers},
  103(1):106--128, 2013.

\bibitem{Bellman1964DifferentialAA}
R.~Bellman, R.~Kalaba, and B.~Kotkin.
\newblock Differential approximation applied to the solution of convolution
  equations.
\newblock {\em Mathematics of Computation}, 18:487--491, 1964.

\bibitem{DISTEFANO19701021}
Néstor Distéfano.
\newblock On alternative representations of time varying viscoelastic
  materials.
\newblock {\em International Journal of Solids and Structures},
  6(7):1021--1033, 1970.

\bibitem{DRAGANOV2010952}
Borislav~R. Draganov.
\newblock Exact estimates of the rate of approximation of convolution
  operators.
\newblock {\em Journal of Approximation Theory}, 162(5):952--979, 2010.

\bibitem{miller2007glauber}
Michael~L Miller, Klaus Reygers, Stephen~J Sanders, and Peter Steinberg.
\newblock Glauber modeling in high-energy nuclear collisions.
\newblock {\em Annu. Rev. Nucl. Part. Sci.}, 57:205--243, 2007.

\bibitem{neal1997monte}
Radford~M Neal.
\newblock Monte carlo implementation of gaussian process models for bayesian
  regression and classification.
\newblock {\em arXiv preprint physics/9701026}, 1997.

\bibitem{gp_robot}
B.~{Bócsi}, D.~{Nguyen-Tuong}, L.~{Csató}, B.~{Schölkopf}, and J.~{Peters}.
\newblock Learning inverse kinematics with structured prediction.
\newblock In {\em 2011 IEEE/RSJ International Conference on Intelligent Robots
  and Systems}, pages 698--703, 2011.

\bibitem{gramacy2008bayesian}
Robert~B Gramacy and Herbert K~H Lee.
\newblock Bayesian treed gaussian process models with an application to
  computer modeling.
\newblock {\em Journal of the American Statistical Association},
  103(483):1119--1130, 2008.

\bibitem{zhang2016big}
Tonglin Zhang and Baijian Yang.
\newblock Big data dimension reduction using pca.
\newblock In {\em 2016 IEEE International Conference on Smart Cloud
  (SmartCloud)}, pages 152--157. IEEE, 2016.

\bibitem{hoffmann2007kernel}
Heiko Hoffmann.
\newblock Kernel pca for novelty detection.
\newblock {\em Pattern recognition}, 40(3):863--874, 2007.

\bibitem{perlibakas2004distance}
Vytautas Perlibakas.
\newblock Distance measures for pca-based face recognition.
\newblock {\em Pattern recognition letters}, 25(6):711--724, 2004.

\bibitem{ding2004k}
Chris Ding and Xiaofeng He.
\newblock K-means clustering via principal component analysis.
\newblock In {\em Proceedings of the twenty-first international conference on
  Machine learning}, page~29, 2004.

\bibitem{owhadi2015bayesian}
Houman Owhadi.
\newblock Bayesian numerical homogenization.
\newblock {\em Multiscale Modeling \& Simulation}, 13(3):812--828, 2015.

\bibitem{griffiths2018introduction}
David~J Griffiths and Darrell~F Schroeter.
\newblock {\em Introduction to quantum mechanics}.
\newblock Cambridge University Press, 2018.

\bibitem{durrer2020cosmic}
Ruth Durrer.
\newblock {\em The cosmic microwave background}.
\newblock Cambridge University Press, 2020.

\bibitem{10.2307/2313748}
Mark Kac.
\newblock Can one hear the shape of a drum?
\newblock {\em The American Mathematical Monthly}, 73(4):1--23, 1966.

\end{thebibliography}

\clearpage

\onecolumngrid

\appendix
\section{Proof of Theorem 1}\label{app:gaussian_theory}
In this appendix, we shall present a formal version of Theorem 1, with rigorous mathematical formulations along with assumptions on the kernel and the function space. Then we make relevant comments on the implications of the assumptions and how it is related to practical scenarios.

{\bf \noindent Theorem 1.}
(Formal) Assume that the coefficients satisfy $-A(\mat{x})= V(\mat{x})$ and $A(\mat{x})\mat{\Sigma}(\mat{x})=\mat{\Sigma}_m(\mat{x})$, and that $\|\sqrt{A(\mat{x})}\|_{L^{\infty}}\|(-\frac{A'^{\otimes2}(\mat{x})}{8A(\mat{x})^{3/2}}+\frac{A''(\mat{x})}{4\sqrt{A(\mat{x})}})\|_{L^{\infty}}\leq c_1$ for some $c_1>0$. In the function space $\mathcal{H}_{c_2}$, the PCA problem of Eq.~(\ref{eq:eigen_pca}) can be approximated by the Schr\"{o}dinger problem of Eq.~(\ref{eq:eigen_se}), in the sense that $|\lambda_i-E_i|\leq O(c_1+c_2)$ for all $i$. Here the function space $\mathcal{H}_{c_2}$ is defined via $\mathcal{H}_{c_2}=\{\phi\in L^2\cap C^3, \|\phi'''\|_{L^{\infty}}\leq c_2 \|\phi\|_{L^{2}}\}$.

\begin{proof}
Since the $i$-th eigenvalue $\lambda_i$ of a positive operator $M$ is related to the min-max formulation of the Rayleigh quotient:
$$
\lambda_{i}=\min_{S: \operatorname{dim}(S)=i } \max _{x \in S,\|x\|=1}\left( M x, x\right),
$$
we only need to bound the differences between Eq.~(\ref{eq:ray_pca}) and (\ref{eq:ray_se}) for $\|\phi\|_{L^2}=1$. 
In the following we aim to build a correspondence between two expressions in $\{\cdots\}$ of Eq.~(\ref{eq:ray_pca}) and (\ref{eq:ray_se}). Firstly we shall introduce an approximation via Taylor's expansion. 

\begin{align}
    & \int_{\mat{y}}d\mat{y}C(\mat{x},\mat{y})\phi(\mat{y})\\
    & =  \int_{\mat{y}}d\mat{y}\sqrt{A(\mat{x})A(\mat{y})}{\rm exp}(-\frac{1}{2}(\mat{y}-\mat{x})^T\mat{\Sigma}^{-1}(\mat{y}-\mat{x}))\phi(\mat{y})\\
    & =  \int_{\mat{r}}d\mat{r}\sqrt{A(\mat{x})A(\mat{x}+\mat{r})}{\rm exp}(-\frac{1}{2}\mat{r}^T\mat{\Sigma}^{-1}\mat{r})\phi(\mat{x}+\mat{r})\\
    \label{eq:only_approx}
    & \approx \int_{\mat{r}}d\mat{r} {\rm exp}(-\frac{1}{2}\mat{r}^T\mat{\Sigma}^{-1}\mat{r}) [A(\mat{x})+\frac{1}{2}\nabla A(\mat{x})\cdot\mat{r})][\phi(\mat{x})+\nabla\phi(\mat{x})\cdot\mat{r}+\frac{1}{2}\mat{r}^T\nabla\nabla^T\phi(\mat{x})\mat{r}]\\
    & = \int_{\mat{r}} d\mat{r}{\rm exp}(-\frac{1}{2}\mat{r}^T\mat{\Sigma}^{-1}\mat{r})(\frac{1}{2}\mat{r}^T(\nabla A(\mat{x})\nabla^T\phi(\mat{x})+A(\mat{x})\nabla\nabla^T\phi(\mat{x}))\mat{r}+\phi(\mat{x})A(\mat{x}))\\
    & = \int_{\mat{r}} d\mat{r}{\rm exp}(-\frac{1}{2}\mat{r}^T\mat{\Sigma}^{-1}\mat{r})(\frac{1}{2}\mat{r}^T(\nabla^T\cdot(A(\mat{x})\nabla\phi\mat(x)))\mat{r}+\phi(\mat{x})A(\mat{x}))\\
    \label{eq:14}
    & = (2\pi)^{\frac{d}{2}}({\rm det}\mat{\Sigma})^\frac{1}{2} (\phi(x)A(\mat{x})+\frac{1}{2}{\rm Tr}(\mat{\Sigma}\nabla^T\cdot(A(\mat{x})\nabla\phi(\mat{x}))))
\end{align}
From Eq.~(A3) to (A4), $A(\mat{x})$ and $\phi(\mat{x})$ are Taylor expanded to first order and second order respectively. From Eq.~(A4) to (A5), odd terms of $\mat{r}$ vanish due to symmetry. From Eq.~(A6) to (A7), we leverage the Gaussian integral:
\begin{equation}
    \int_{\mat{r}}d\mat{r}{\rm exp}(-\frac{1}{2}\mat{r}^T\mat{\Sigma}^{-1}\mat{r})(\mat{r}^T\mat{B}\mat{r})=(2\pi)^{\frac{d}{2}}({\rm det}\mat{\Sigma})^{\frac{1}{2}}{\rm Tr}(\mat{\Sigma}\mat{B})
\end{equation}
We insert the trace term back to Eq.~(\ref{eq:ray_pca}) and invoke integration by parts:
\begin{align}
    &\int_{\mat{x}}d\mat{x}\phi(\mat{x})\cdot {\rm Tr}(\mat{\Sigma}\nabla^T\cdot(A(\mat{x})\nabla\phi(\mat{x})))\\
    & = {\rm Tr}(\mat{\Sigma}\int_{\mat{x}}d\mat{x}\phi(\mat{x})\nabla^T\cdot(A(\mat{x})\nabla\phi(\mat{x})))\\
    & = {\rm Tr}(\mat{\Sigma}\int_{\mat{x}}d\mat{x}[\nabla^T\cdot(\phi(\mat{x})A(\mat{x})\nabla\phi(\mat{x}))-\nabla\phi(\mat{x})(A(\mat{x})\nabla\phi(\mat{x})))^T])\\
    & = -{\rm Tr}(\mat{\Sigma} \int_{\mat{x}}d\mat{x} A(\mat{x})\nabla\phi(\mat{x})\nabla^T\phi(\mat{x}))\\
    \label{eq:19}
    & = -\int_{\mat{x}}d\mat{x} A(\mat{x})\nabla^T\phi(\mat{x})\mat{\Sigma}\nabla\phi(\mat{x})
\end{align}
By inserting Eq.~(\ref{eq:14}) and ~(\ref{eq:19}) back to Eq.~(\ref{eq:ray_pca}) we have
\begin{align}\label{eq:ray_final_pca}
    R_{pca}(C,\phi)=-(2\pi)^{\frac{d}{2}}({\rm det}\mat{\Sigma})^{\frac{1}{2}}\int_{\mat{x}}d\mat{x} (-A(\mat{x})\phi^2(\mat{x})+\nabla^T\phi(\mat{x})(\frac{1}{2}A(\mat{x})\mat{\Sigma})\nabla\phi(\mat{x}))
\end{align}
After integration by parts, Eq.~(\ref{eq:eigen_se}) now becomes
\begin{align}\label{eq:ray_final_se}
    R_{se}(\hat{H},\psi)=\int_{\mat{x}}d\mat{x} (V(\mat{x})\psi^2(\mat{x})+\nabla^T\psi(\mat{x})\frac{1}{2}\mat{\Sigma}_m\nabla\psi(\mat{x}))
\end{align}
To equate Eq.~(\ref{eq:ray_final_pca}) and (\ref{eq:ray_final_se}) (ignoring the constant factor in Eq.~(\ref{eq:ray_final_pca})), we only need
Eq.~(\ref{eq:equivalence}).

Therefore, we only need to give an estimate of the approximation  in order to obtain a qualitative bound on the error.

In fact, the Taylor's residue in Eq.~(\ref{eq:only_approx}) in approximation of $\sqrt{A(\mat{x})A(\mat{x}+\mat{r})}$ is $\sqrt{A(\mat{x})}(-\frac{A'^{\otimes2}(\mat{x}+\epsilon \mat{r})}{8A(\mat{x}+\epsilon \mat{r})^{3/2}}+\frac{A''(\mat{x}+\epsilon \mat{r})}{4\sqrt{A(\mat{x}+\epsilon\mat{r})}})\cdot \mat{r}^{\otimes 2}$. Similarly, the Taylor's residue in line Eq.~(\ref{eq:only_approx}) in approximation of $\phi(\mat{x}+ \mat{r})$ is $\frac{1}{6} \phi'''(\mat{x}+\epsilon \mat{r}) \cdot \mat{r}^{\otimes 3}$.
Thus we can bound the approximation error as:
\begin{align*}
    & \int_{\mat{r}} d\mat{r}{\rm exp}(-\frac{1}{2}\mat{r}^T\mat{\Sigma}^{-1}\mat{r})\{|c_1 \mat{r}^{\otimes 2}||[\phi(\mat{x})+\nabla\phi(\mat{x})\cdot\mat{r}+\frac{1}{2}\mat{r}^T\nabla\nabla^T\phi(\mat{x})\mat{r}]|\\&+|[A(\mat{x})+\frac{1}{2}\nabla A(\mat{x})\cdot\mat{r}]|\frac{1}{6} \|\phi'''\|_{L^{\infty}} | \mat{r}^{\otimes 3}|+|c_1 \mat{r}^{\otimes 2}|\frac{1}{6} \|\phi'''\|_{L^{\infty}}  |\mat{r}^{\otimes 3}|\}
\end{align*}
Evaluating the Gaussian integrals, we obtain the desired error bound of $O(c_1+c_2)$.

\end{proof}

To avoid confusion, we would like to provide a few remarks here:

(1) The only approximation made in the proof of Theorem \ref{thm:1} lies in Eq.~(\ref{eq:only_approx}) where $A(\mat{x})$ and $\phi(\mat{x})$ are taylor expanded to first-derivative and second-derivative, respectively. The underlying assumption is that (a) the correlation lengths (eigenvalues of $\mat{\Sigma}$) should be much smaller than the size of whole integral domain, so that the Gaussian factor decays fast enough before Taylor expansion fails; (b) the variation of $A(\mat{x})$ is milder than variation of $\phi(\mat{x})$, or $\nabla\nabla^TA(\mat{x})/A(\mat{x})\ll \nabla\nabla^T\phi(\mat{x})/\phi(\mat{x})$. 

(2) The assumptions on the coefficients and the function space might appear strange.  The purpose is to impose a control of higher-order derivatives, so as to make the approximation via Tayler’s expansion valid. In practical implementations, we shall relax the constraint and search for eigenfunctions in $L^2$. For the numerical experiments and its lower eigenmodes, the functions are with low oscillation with associated constants $c_1$, $c_2$ small, so the approximation is indeed accurate. There shall be more oscillations for high eigenmode functions, in the regime beyond the applicability of our theorem, which explains why our approximation is more valid and mostly applicable to the approximation of the first few eigenmodes.
Unfortunately, we don't know a priori how many lower eigenmodes are applicable for this theorem. The statement of the theorem can only serve as a posteriori criteria in the numerical experiments  since the desired a priori estimates on the first few eigenmodes are not available.

(3) Although we have assumed $C(\mat{x},\mat{y})$ as Gaussian, our proof also works for any function that can be decomposed into a linear combination of Gaussian radial basis function. The equivalence between PCA and SE remains valid when the covariance function $C(\mat{x},\mat{y})$ is of the form of an exponential function, which is reminiscent of many condensed matter systems in physics, e.g. Ising model. Please refer to Appendix B and C for details. 

(4) The corresponding relations Eq.~(\ref{eq:equivalence}) are particularly interesting in terms of physics. Normally one locates the system at disposal on a compact support $\Omega_c\in\Omega$, i.e. $A(\mat{x})>0$ for $x\in\Omega_s$ while $A(\mat{x})=0$ for $x\in\Omega/\Omega_s$, $V(\mat{x})=-A(\mat{x})$ implies that the quantum particle is trapped in a potential well on $\Omega_s$. Further, if $A(\mat{x})$ can be approximated as a quadratic form around the global maxima, the Schr\"{o}dinger equation is then readily identified as a well-studied system in physics: the quantum harmonic oscillator, the properties of which are summarized in Appendix E.

\section{Generalization of Theorem 1 to Exponential Correlations}\label{app:exponential_theory}

In the physics community, we are often concerned with covariance fields of the following type:
\begin{equation}
    C(x,y)=\sqrt{A(\mat{x})A(\mat{y})}{\rm exp}(-|\mat{\Sigma}^{-\frac{1}{2}}(\mat{y}-\mat{x})|)/Z_\mat{\Sigma},
\end{equation}
where $\mat{\Sigma}$ is a positive definite function, $c_{d-1}$ is the area surface of the unit sphere in $R^d$, $Z_\mat{\Sigma}:=(d-1)! c_{d-1}({\rm det}\mat{\Sigma})^{\frac{1}{2}} $ is the normalization factor of the covariance function, and $|\cdot|$ refers to vector 2-norm.

Similar to the Gaussian case, we shall use Sch\"{o}dinger equation to perform PCA approximations. Here we identify when the PCA problem described by Eq.~(\ref{eq:eigen_pca}) can be well approximated by the Schr\"{o}dinger problem described by Eq.~(\ref{eq:eigen_se}) as in Theorem 1.

\begin{theorem}\label{thm:2}(Informal)
The PCA problem described by Eq.~(\ref{eq:eigen_pca}) can be approximated by the Schr\"{o}dinger problem described by Eq.~(\ref{eq:eigen_se}). i.e. two systems have the same equation hence same eigenvalues and eigenmodes, upto a second-order approximation provided that these quantities are equal:
\begin{eqnarray}\label{eq:equivalence1}
\begin{aligned}
    \phi(\mat{x})&\Longleftrightarrow\psi(\mat{x})\\ -A(\mat{x})&\Longleftrightarrow V(\mat{x})\\ (d+1)A(\mat{x})\mat{\Sigma}&\Longleftrightarrow\mat{\Sigma}_m
\end{aligned}
\end{eqnarray}

\end{theorem}

\begin{proof}
In the following we aim to build a correspondence between two expressions in $\{\cdots\}$ of Eq.~(\ref{eq:ray_pca}) and (\ref{eq:ray_se}).

\begin{align}
    & \int_{\mat{y}}d\mat{y}C(\mat{x},\mat{y})\phi(\mat{y})\\
    & =  \int_{\mat{y}}d\mat{y}\sqrt{A(\mat{x})A(\mat{y})}{\rm exp}(-|\mat{\Sigma}^{-\frac{1}{2}}(\mat{y}-\mat{x})|)\phi(\mat{y})/Z_\mat{\Sigma}\\
    & =  \int_{\mat{r}}d\mat{r}\sqrt{A(\mat{x})A(\mat{x}+\mat{r})}{\rm exp}(-|\mat{\Sigma}^{-\frac{1}{2}}(\mat{r})|)\phi(\mat{x}+\mat{r})/Z_\mat{\Sigma}\\
    \label{eq:only_approx1}
    & \approx \int_{\mat{r}}d\mat{r} {\rm exp}(-|\mat{\Sigma}^{-\frac{1}{2}}(\mat{r})|) [A(\mat{x})+\frac{1}{2}\nabla A(\mat{x})\cdot\mat{r})][\phi(\mat{x})+\nabla\phi(\mat{x})\cdot\mat{r}+\frac{1}{2}\mat{r}^T\nabla\nabla^T\phi(\mat{x})\mat{r}]/Z_\mat{\Sigma}\\
    & = \int_{\mat{r}} d\mat{r}{\rm exp}(-|\mat{\Sigma}^{-\frac{1}{2}}(\mat{r})|)(\frac{1}{2}\mat{r}^T(\nabla A(\mat{x})\nabla^T\phi(\mat{x})+A(\mat{x})\nabla\nabla^T\phi(\mat{x}))\mat{r}+\phi(\mat{x})A(\mat{x}))/Z_\mat{\Sigma}\\
    & = \int_{\mat{r}} d\mat{r}{\rm exp}(-|\mat{\Sigma}^{-\frac{1}{2}}(\mat{r})|)(\frac{1}{2}\mat{r}^T(\nabla^T\cdot(A(\mat{x})\nabla\phi\mat(x)))\mat{r}+\phi(\mat{x})A(\mat{x}))/Z_\mat{\Sigma}\\
    \label{eq:141}
    & = \phi(x)A(\mat{x})+\frac{1}{2}(d+1){\rm Tr}(\mat{\Sigma}^{\frac{1}{2}}\nabla^T\cdot(A(\mat{x})\nabla\phi(\mat{x}))\mat{\Sigma}^{\frac{1}{2}})
\end{align}
From Eq.~(B5) to (B6), $A(\mat{x})$ and $\phi(\mat{x})$ are Taylor expanded to first order and second order respectively. From Eq.~(B6) to (B7), odd terms of $\mat{r}$ vanish due to symmetry. From Eq.~(B8) to (B9), we leverage the Exponential integral:
\begin{equation}
    \int_{\mat{r}}d\mat{r}{\rm exp}(-|\mat{\Sigma}^{-\frac{1}{2}}(\mat{r})|)(\mat{r}^T\mat{B}\mat{r})/Z_\mat{\Sigma}=(d+1){\rm Tr}(\mat{\Sigma}^{\frac{1}{2}}\mat{B}\mat{\Sigma}^{\frac{1}{2}})
\end{equation}
We insert the trace term back to Eq.~(\ref{eq:ray_pca}) and invoke integration by parts:
\begin{align}
    &\int_{\mat{x}}d\mat{x}\phi(\mat{x})\cdot {\rm Tr}(\mat{\Sigma}^{\frac{1}{2}}\nabla^T\cdot(A(\mat{x})\nabla\phi(\mat{x}))\mat{\Sigma}^{\frac{1}{2}})\\
    & = {\rm Tr}(\mat{\Sigma}^{\frac{1}{2}}\int_{\mat{x}}d\mat{x}\phi(\mat{x})\nabla^T\cdot(A(\mat{x})\nabla\phi(\mat{x}))\mat{\Sigma}^{\frac{1}{2}})\\
    & = {\rm Tr}(\mat{\Sigma}^{\frac{1}{2}}\int_{\mat{x}}d\mat{x}[\nabla^T\cdot(\phi(\mat{x})A(\mat{x})\nabla\phi(\mat{x}))-\nabla\phi(\mat{x})(A(\mat{x})\nabla\phi(\mat{x})))^T]\mat{\Sigma}^{\frac{1}{2}})\\
    & = -{\rm Tr}( \mat{\Sigma}^{\frac{1}{2}}\int_{\mat{x}}d\mat{x} A(\mat{x})\nabla\phi(\mat{x})\nabla^T\phi(\mat{x})\mat{\Sigma}^{\frac{1}{2}})\\
    \label{eq:191}
    & = -\int_{\mat{x}}d\mat{x} A(\mat{x})\nabla^T\phi(\mat{x})\mat{\Sigma}\nabla\phi(\mat{x})
\end{align}
By inserting Eq.~(\ref{eq:141}) and ~(\ref{eq:191}) back to Eq.~(\ref{eq:ray_pca}) we have
\begin{align}\label{eq:ray_final_pca1}
    R_{pca}(C,\phi)=-\int_{\mat{x}}d\mat{x} (-A(\mat{x})\phi^2(\mat{x})+\frac{1}{2}(d+1)\nabla^T\phi(\mat{x})(A(\mat{x})\mat{\Sigma}\nabla\phi(\mat{x}))
\end{align}
After integration by parts, Eq.~(\ref{eq:eigen_se}) now becomes
\begin{align}\label{eq:ray_final_se1}
    R_{se}(\hat{H},\psi)=\int_{\mat{x}}d\mat{x} (V(\mat{x})\psi^2(\mat{x})+\frac{1}{2}\nabla^T\psi(\mat{x})\mat{\Sigma}_m\nabla\psi(\mat{x}))
\end{align}
To equate Eq.~(\ref{eq:ray_final_pca1}) and (\ref{eq:ray_final_se1}), we only need
Eq.~(\ref{eq:equivalence1}).
\end{proof}

\section{Experimental results for Laplace kernel}\label{app:exponential_exp}

\begin{figure}[htbp]
    \centering
    \includegraphics[width=0.4\linewidth]{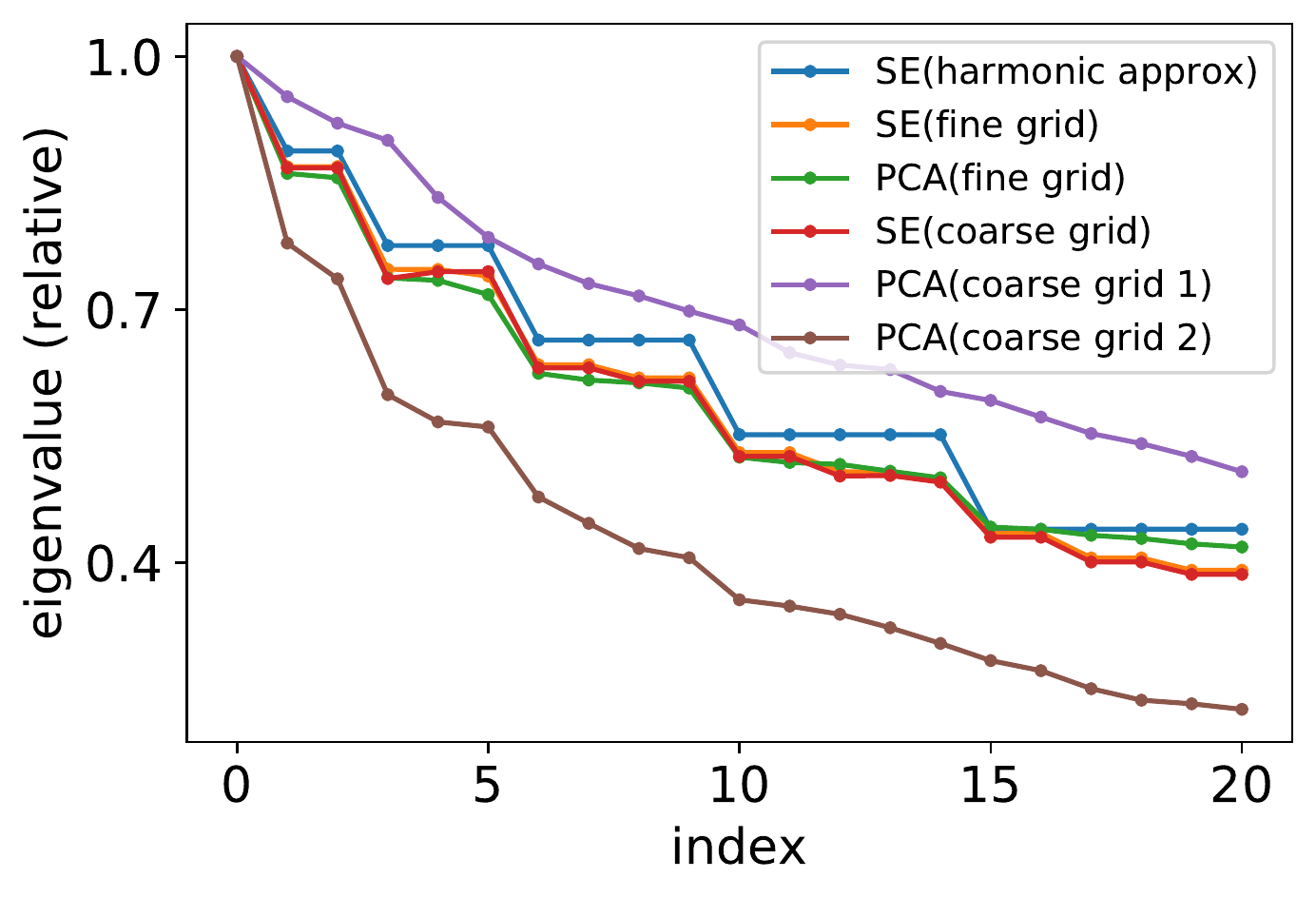}
    \caption{Exponential correlation example: the (relative) eigenvalues of first 21 modes in 5 different cases. Note: (1) On the fine grid, Schr\"{o}dinger PCA and PCA obtain similar eigenvalues; (2) On the coarse grid, Schr\"{o}dinger PCA still works while PCA fails; (3) Harmonic approximation is nearly exact for first 15 eigenvalues;}
    \label{fig:eigenvalues_lap}
\end{figure}

\begin{figure}[t]
    \centering
    \includegraphics[width=0.7\linewidth]{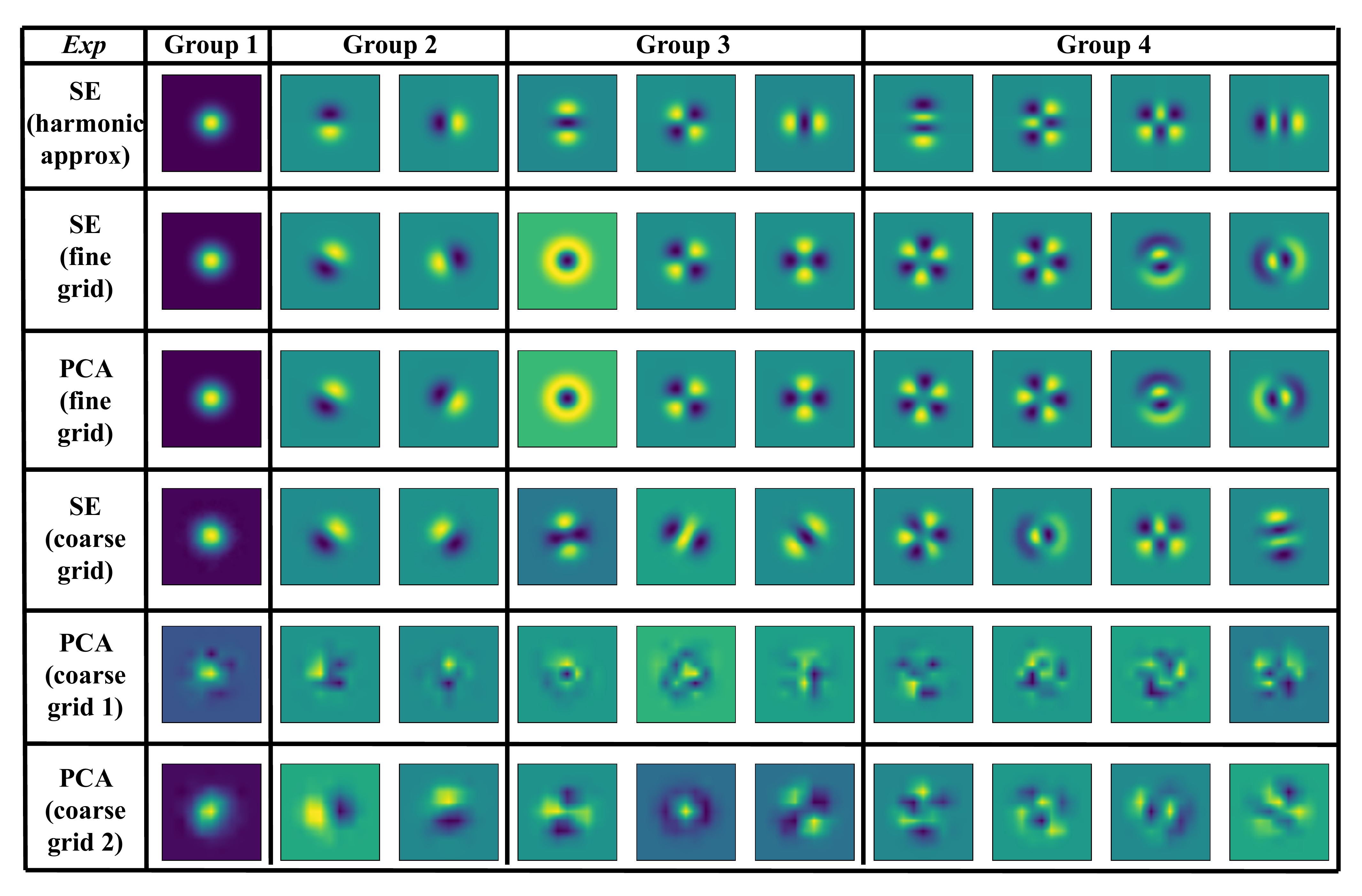}
    \caption{The exponential correlation example: eigenmodes of first 10 modes in 5 different cases. (1) On the fine grid, Schr\"{o}dinger PCA and PCA obtain similar eigenmodes; (2) Harmonic approximation is nearly exact for first 10 eigenmodes; (3) On the coarse grid, Schr\"{o}dinger PCA still works while PCA fails.}
    \label{fig:modes_lap}
\end{figure}

\section{Generation of Gaussian process}\label{app:generate_gp}
In this section, we will discuss the generation of the rectangular Gaussian process (GP) in our numerical experiments. We simply set $\mat{\Sigma}=\sigma^2\mat{I}$  in Eq. (\ref{eq:grf}).\par
The generation process contains two stages: (1) we add Gaussian filtering to match the correlation kernel, i.e. ${\rm exp}(-\frac{(\mat{y}-\mat{x})^T(\mat{y}-\mat{x})}{2\sigma^2})$ in Eq. (\ref{eq:grf}). (2) We re-scale the field dependent on $\mat{x}$ to match the variance $A(\mat{x})$ and magnitude of covariance $\sqrt{A(\mat{x})A(\mat{y})}$ in Eq. (\ref{eq:grf}).

\subsection{Covariance Generation}
Since the two-dimensional GP lies on grid$[-50,50]\times[-50,50]$(size $101\times101$), we start from generating a $101\times101$ matrix with each component independently sampled from a standard normal distribution.

As we mentioned, an isotropic correlation length $\sigma$ is introduced in our experiment. The width of the Gaussian filter, denoted as $\sigma_0$, is set to match the correlation length $\sigma$. It is easy to see that $\sigma=\sqrt{2}\sigma_0$. \par
For the standard Gaussian filtering algorithm, a truncation along each direction is introduced for efficient calculation. In our numerical experiment, we set this truncation level at $5\sigma$\footnote{This truncation results in a $\mathcal{O}(10^{-7})$ relative error level.}.
\begin{equation}
    C(\mat{x},\mat{y})=C_{ideal}(\mat{x},\mat{y})\eta(5\sigma-|x_1-y_1|)\eta(5\sigma-|x_2-y_2|)
    \label{eq:trunc}
\end{equation}
Here the $\eta(\cdot)$ is the step function. 
\begin{equation}
    \eta(x) = \begin{cases}
    0,x<0;\\
    1,x>0;\\
    \frac{1}{2},x=0.
    \end{cases}
    \label{eq:step}
\end{equation}
With this truncation, the Gaussian filtering can be achieved by convolution with a Gaussian kernel $G_\sigma$.

\subsection{Modifying magnitudes}
Then, due to the scaling property of gaussian distribution, we design a window function $A'(\mat{x})$ to modulate the original GP. Field values at point $\mat{x}$ are multiplied by $\sqrt{A(\mat{x})}$.

\subsection{The overall algorithm}
Based on the above discussion, the overall algorithm to generate $n$ number of samples GP with pre-defined variance $A(\mat{x})$ and the gaussian-like covariance $C(\mat{x},\mat{y})$ with pre-defined isotropic correlation length $\sigma$ is as Algorithm \ref{alg:genGRF}.
\begin{algorithm}[]
\SetAlgoLined
\KwResult{sample size $n$, grid infomation $[-50,50]\times[-50,50]\sim 101\times101$ \;variance definition $A(\mat{x})$, isotropic correlation length $\sigma$}
Determine $A'(\mat{x})$ with given $\sigma$ and $A(\mat{x})$\;
 \For{$i=1,2,\cdots,n$}{
    (1) Get the initial GP sample $\phi^0_i(\mat{x})$,$\phi^0_i(\mat{x}_{j,k})\sim \text{i.i.d.} N(0,1^2)$ for $j,k \in [-50,50]$, here $N(0,1^2)$ stands for the standard normal distribution\;
    (2) Add correlations via Gaussian filtering at $5\sigma$ truncation level. $\phi_i(\mat{x})=\phi_i^0(\mat{x})\otimes G_{\sigma}$.  \;
    (3) Obtain $\psi_i(\mat{x})$ via modulating $\phi_i(\mat{x})$ through $\psi_i(\mat{x}_{j,k})=A'(\mat{x}_{j,k})\cdot\phi_i(\mat{x}_{j,k})$.
    }
    \KwResults{GP samples $\{\psi_i(\mat{x})\}_{i=1}^{n}$}
 \caption{Generation of Gaussian Process}
 \label{alg:genGRF}
\end{algorithm}

\section{Properties of quantum harmonic oscillator}\label{app:harmonic}

The standard equation (in physics) of a 1-dimensional harmonic oscillator is written as\footnote{We have set the Planck consant $\hbar=1$.}:
\begin{equation}
    -\frac{1}{2m}\nabla^2\psi(x)+\frac{1}{2}m\omega^2x^2\psi(x)=E\psi(x)
\end{equation}
whose eigenstates $\psi_n(\mat{x})$ and eigenenergies $E_n$ are:
\begin{equation}
    \psi_n(x)=\frac{1}{\sqrt{2^nn!}}(\frac{m\omega}{\pi})^\frac{1}{4}e^{-\frac{m\omega x^2}{2}}H_n(\sqrt{m\omega}x), E_n=(n+\frac{1}{2})\omega\quad (n=0,1,2,\cdots)
\end{equation}
where $H_n(z)=(-1)^ne^{z^2}\frac{d^n}{dz^n}(e^{-z^2})$ are Hermite polynomials. For the 2-dimensional harmonic oscillator in our case, it can be decoupled into two independent oscillators along $x_1$ and $x_2$ respectively. The total energy is equal to the sum of energy in both directions $E=E_1+E_2$, and the total wavefunction is equal to the product of wavefunction in both directions $\psi(x_1,x_2)=\psi_1(x_1)\psi_2(x_2)$.

\section{Generation of Icosahedron mesh}\label{app:mesh}

There are many ways to build meshes on a sphere: UV sphere, normalized cube, spherified cube and icosahedron. UV sphere has singularities around two poles, and normalized cube and spherified cube do not sustain enough rational symmetry. Given the above considerations, we choose icosahedron mesh to discretize the sphere.

Our mesh starts from an icosahedron (20 faces, 30 edges, and 12 vertices). Each face of icosahedron is an equidistant triangle. \textit{Subdivision} is to partition one triangle into four smaller equidistant triangles, as shown in Figure \ref{fig:icosahedron}. Then middle points are re-scaled to project to the surface of the sphere. A larger number of subdivisions generates finer meshes, shown in Figure \ref{fig:icolevel}. Both Figure \ref{fig:icosahedron} and \ref{fig:icolevel} courtesy to \url{http://blog.andreaskahler.com/2009/06/creating-icosphere-mesh-in-code.html}.

\begin{figure}[htbp]
    \centering
    \includegraphics[width=0.6\linewidth]{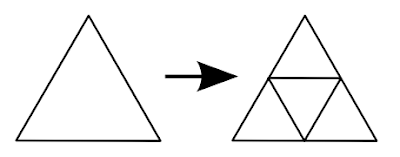}
    \caption{Divide one equidistant triangle to four smaller equidistant triangles.}
    \label{fig:icosahedron}
\end{figure}

\begin{figure}[htbp]
    \centering
    \includegraphics[width=0.6\linewidth]{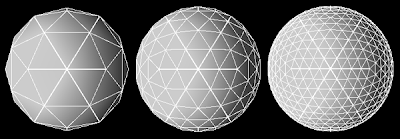}
    \caption{More subdivisions generate finer meshes.}
    \label{fig:icolevel}
\end{figure}

After the icosahedron mesh is generated, corresponding degree matrix $\mat{D}$, adjacency matrix $\mat{A}$ and Laplacian matrix $\mat{L}=\mat{D}-\mat{A}$ can be computed. In the global climate example (section \ref{sec:exp_climate}), we replace the Laplacian operator $\nabla^2$ with the Laplacian matrix $\mat{L}$ in the Schr\"{o}dinger equation, i.e.
\begin{equation}
    \nabla^2\to -\mat{L}d^2
\end{equation}
where $d$ is the averaged length of edges.

\end{document}